\documentclass[11pt]{article}
\usepackage[utf8]{inputenc}
\usepackage[T1]{fontenc}   
\usepackage[hidelinks]{hyperref}
\usepackage{url}             
\usepackage{microtype}     
\usepackage{mathtools}
\usepackage{xcolor}   
\usepackage{graphicx}    
\usepackage{physics}
\usepackage{amsmath}
\usepackage{amsthm}
\usepackage{multicol}
\usepackage{cleveref}
\usepackage{dirtytalk}
\usepackage{enumitem}

\usepackage[style=ieee-alphabetic]{biblatex}
\usepackage{libertinus}
\usepackage[libertine]{newtxmath}

\usepackage[
    top=2.5cm,
    bottom=2.5cm,
    left=1.6cm,
    right=1.6cm,
]{geometry}

\usepackage[classfont = bold,
            langfont  = caps,
            funcfont  = roman]{complexity}

\newlang{\UNSAT}{Unsat}
\newlang{\TAUT}{Taut}

\newcommand{\SOT}{\mathsf{S^1_2}}
\renewcommand{\PV}{\mathsf{PV}}
\newcommand{\PVO}{\mathsf{PV_1}}
\newcommand{\SOTEXP}{\mathsf{S^1_2} + \mathsf{1\textsf{-}EXP}}
\newcommand{\APC}{\mathsf{APC_1}}
\newcommand{\dWPHP}{\operatorname{dWPHP}}

\newcommand{\trans}[1]{||#1||}

\newcommand{\EF}{\mathsf{EF}}
\newcommand{\Gent}{\mathsf{G}}
\newcommand{\GentS}{\mathsf{G^*}}
\newcommand{\GOS}{\mathsf{G^*_1}}
\newcommand{\IPS}{\mathsf{IPS}}
\newcommand{\iEF}{\mathsf{iEF}}
\newcommand{\imp}{\mathsf{i}}
\renewcommand{\SC}{\mathsf{SC}}

\newcommand{\sound}{\operatorname{Sound}}
\DeclareMathOperator*{\foo}{Pr}
\def\<{\left<}
\def\>{\right>}
\newcommand{\setsize}{\operatorname{Size}}
\def\tt{\operatorname{CorrectFracTT}}
\def\hardav{\operatorname{Hard}_{1/4}^\text{A}}
\def\rra{\twoheadrightarrow}
\def\OEXP{\mathsf{1\textsf{-}EXP}}
\def\GE{\mathsf{G_{EXP}}}
\def\sigmaref{\Sigma^q_1\text{-}\operatorname{Ref}}
\def\LB{\ttable}
\def\refprop{\operatorname{Ref}}

\theoremstyle{definition}
\newtheorem{definition}{Definition}[section]

\theoremstyle{plain}
\newtheorem{theorem}[definition]{Theorem}
\newtheorem{lemma}[definition]{Lemma}
\newtheorem{corollary}[definition]{Corollary}
\newtheorem{proposition}[definition]{Proposition}

\newtheorem*{theorem*}{Theorem}
\newtheorem*{corollary*}{Corollary}

\theoremstyle{remark}

\DeclareMathOperator*{\ttt}{tt}
\newcommand{\ttable}{\ttt^{\operatorname{avg}}}
\newcommand{\EmilU}{\mathbin{\dot{\cup}}}

\newcommand{\bN}{\mathbb{N}}
\newcommand{\bbN}{\mathbb{N}}
\newcommand{\bF}{\mathbb{F}}

\newcommand{\calC}{\mathcal{C}}

\newcommand{\size}[2]{\mathsf{size}_{#1}(#2)}

\newcommand{\Log}{\operatorname{Log}}
\newcommand{\LogLog}{\operatorname{LogLog}}

\newcommand{\email}[1]{\href{mailto:#1}{\textsf{#1}}}

\title{ From Proof Complexity to Circuit Complexity \\via Interactive Protocols}

\author{
    Noel Arteche\footnote{Lund University and University of Copenhagen, \email{noel.arteche@cs.lth.se}}
    \and
    Erfan Khaniki\footnote{Institute of Mathematics of the Czech Academy of Sciences, \email{e.khaniki@gmail.com}}
    \and
    Ján Pich\footnote{University of Oxford, \email{jan.pich@cs.ox.ac.uk}}
    \and
    Rahul Santhanam\footnote{University of Oxford, \email{rahul.santhanam@cs.ox.ac.uk}}
}

\date{}

\begin{document}
\maketitle
\abstract{
Folklore in complexity theory suspects that circuit lower bounds against $\NC^1$ or $\P/\poly$, currently out of reach, are a necessary step towards proving strong proof complexity lower bounds for systems like Frege or Extended Frege. Establishing such a connection formally, however, is already daunting, as it would imply the breakthrough separation $\NEXP \not\subseteq \P/\poly$, as recently observed by \citeauthor{Pseopt} \cite{Pseopt}.

We show such a connection conditionally for the Implicit Extended Frege proof system ($\iEF$) introduced by \citeauthor{Kra04} \cite{Kra04}, capable of formalizing most of contemporary complexity theory. In particular, we show that if $\iEF$ proves efficiently the standard derandomization assumption that a concrete Boolean function is hard on average for subexponential-size circuits, then any superpolynomial lower bound on the length of $\iEF$ proofs implies $\#\P \not\subseteq \FP/\poly$ (which would in turn imply, for example, $\PSPACE \not\subseteq \P/\poly$). Our proof exploits the formalization inside $\iEF$ of the soundness of the sum-check protocol of Lund, Fortnow, Karloff, and Nisan \cite{LFKN92-SumCheck}. This has consequences for the self-provability of circuit upper bounds in $\iEF$. Interestingly, further improving our result seems to require progress in constructing interactive proof systems with more efficient provers.
}

\newpage

\section{Introduction}
\label{sec:intro}

At a high level, both circuit complexity and proof complexity can be thought of as an approach towards the $\P$ versus $\NP$ question. The circuit complexity program, which met with considerable success in the 1980s, tries to prove lower bounds against gradually larger circuit classes, hoping to eventually show $\NP \not\subseteq \P/\poly$. Proof complexity, often identified with the so-called Cook-Reckhow program, intends to show $\NP \neq \coNP$ and, in turn, $\P \neq \NP$, by proving lower bounds against gradually more powerful proof systems for propositional logic.

While both enterprises share the motivation to study \emph{concrete} computational models of increasing power hoping to build up techniques to attack the long-sought separations, there exist notable differences. Circuit complexity looks at deterministic models of computation, while proof complexity deals with proof systems, which are inherently non-deterministic. Furthermore, while circuit complexity has a clear end-goal (lower bounds against general Boolean circuits), it remains wide open whether the Cook-Reckhow program can be realized even in principle. It is not known whether lower bounds against strong systems like Extended Frege can imply lower bounds for every other system and, as such, one could potentially keep proving lower bounds for ever-stronger systems without ever settling whether $\NP \neq \coNP$.

The parallels between circuit complexity and proof complexity are made clearer by Frege systems. For each circuit complexity class $\calC$, one can define the proof system $\calC$-Frege, in which proof lines are restricted to be circuits from $\calC$. In this setting strong systems like Frege and Extended Frege correspond to $\NC^1$-Frege and $\P/\poly$-Frege, respectively, and thus the natural question arises: Can we turn explicit lower bounds for $\calC$ circuits into lower bounds for $\calC$-Frege systems, and vice versa?

While the question is essentially open, work on weaker systems and circuit classes has proven successful. In one direction, the method of feasible interpolation \cite{Kra94-PHP, Razb95, Kra97} (see \cite[§17.9.1]{krajicekBOOK} for the history of the method) has been extensively applied to obtain proof complexity lower bounds. The framework of feasible interpolation formalizes the idea of extracting computational content from proofs: given short proofs in a given system, one can extract a small Boolean circuit in some restricted classes for a related interpolant function. Contrapositively, circuit lower bounds for such functions (often coming from unconditional results such as lower bounds against monotone circuits \cite{Razb85, And85, AB87}), turn into lower bounds for proofs systems like Resolution \cite{Kra97} or Cutting Planes \cite{Pud97} (and conditionally for other systems, such as Polynomial Calculus or Sum-of-Squares \cite{Hak20}). Unfortunately, this connection breaks for stronger proof systems: already $\AC^0$-Frege and $\TC^0$-Frege are known to lack feasible interpolation properties\footnote{Some of these systems are known to admit some form of interpolation by stronger computational models, see e.g. \cite{Pud20, DR23}, but we are interested in Boolean circuits.} under standard cryptographic hardness assumptions \cite{KP98, BPR1997, BDGMP04}, and this holds even if we allow feasible interpolation by quantum circuits \cite{ ACG24}.

In the other direction (circuit complexity from proof complexity), the theory of lifting has unveiled deep connections between proofs, circuits and communication protocols. Here, so-called query-to-communication lifting theorems translate query complexity lower bounds (corresponding to weak systems, like Resolution) into communication complexity lower bounds (e.g.\ \cite{RMK97, LMMPZ22}). The latter provide restricted circuit lower bounds, such as for monotone circuits (see e.g.\ \cite{GGKS18, dRMNPRV20, dRGR22} and references therein). It is, however, not known how to derive non-monotone lower bounds for unrestricted Boolean circuits by lifting proof complexity lower bounds.

For proper Frege systems, the connection has worked mostly in one direction, from circuits to proofs, particularly at the level of techniques. The method of random restrictions and the celebrated switching lemmas used to show constant-depth circuit lower bounds in the 1980s \cite{FSS-AC0, Ajt83, Has86} were successfully transferred into $\AC^0$-Frege lower bounds shortly after \cite{Ajt94-PHP, BPU92-PHP,  Kra94-PHP,  jointPHP, PBI93, KPW95-AC0}. This suggests that understanding what makes proof lines large might be necessary to understand why proofs are long. Intriguingly, understanding the proof lines alone does not seem to suffice: the $\AC^0[p]$ lower bounds of Razborov and Smolensky \cite{Razb87, Smo87} are yet to be successfuly translated to proof complexity, with lower bounds for $\AC^0[2]$-Frege being one of the prominent frontier problems in the field.

The current situation seems to suggest that in order to make progress towards proof complexity lower bounds, it is \emph{necessary} (though seemingly not sufficient) to first obtain strong enough circuit lower bounds. In particular, under this folklore belief, circuit lower bounds against $\NC^1$ or $\P/\poly$, currently out or reach, would be a necessary step towards proving strong proof complexity lower bounds for systems like Frege or Extended Frege. However, the suspicion remains unproven, and no generic way of deriving explicit circuit lower bounds for unrestricted Boolean circuits from proof complexity lower bounds for concrete propositional proof systems has been discovered\footnote{We note that the issue lies in establishing such a connection for a \emph{concrete} system. Of course, the statement \say{there is a proof system $S$ such that if $S$ is not polynomially bounded, then $\P \neq \NP$} is true: if $\NP = \coNP$ the implication is vacuously true by taking a polynomially bounded proof system; if $\NP \neq \coNP$, then $\P \neq \NP$ and thus the statement holds for any proof system. It would be dramatically different to obtain such a connection for a concrete system.}.

The first result giving such a connection under relatively conventional assumptions which are presumably weaker than the conclusion of the connection itself was presented recently by \citeauthor{Pseopt} \cite{Pseopt}. Specifically, they showed that any superpolynomial lower bound on the length of tautologies in the Extended Frege system $\EF$ implies $\NP\not\subseteq\Ppoly$ assuming hypotheses (I) and (II) below:

\begin{itemize}
\item[\textsc{I}.] (Provable circuit lower bound.) $\EF$ proves efficiently that a concrete Boolean function in $\E$ is average-case hard for subexponential-size circuits.

\item[\textsc{II}.] (Provable reduction of OWFs to $\P\neq\NP$.) $\EF$ proves efficiently that a polynomial-time function transforms circuits breaking one-way functions into circuits solving $\SAT$.
\end{itemize}

We remark that Hypothesis I above presupposes $\E\not\subseteq\Ppoly$, which is however believed to be a significantly weaker statement than $\NP\not\subseteq\Ppoly$. Alternatively, Hypotheses I and II can be replaced by a single assumption on the feasible provability of the existence of anticheckers in $\EF$. These results remain valid even if we replace $\EF$ by an essentially arbitrary proof system simulating $\EF$.

Crucially, improving this and related results by dropping the hypotheses is surprisingly daunting. As noted by \citeauthor{Pseopt} \cite[Prop. 1]{Pseopt}, if one unconditionally establishes the implication \say{if $S$ is not polynomially bounded, then $\NP \not\subseteq \Ppoly$} for a concrete proof system $S$, then the breakthrough separation $\NP \not\subseteq \SIZE[n^k]$, for every fixed $k$ (and $\NEXP\not\subseteq\P/\poly$) follows!

In short, proving a formal connection between proof complexity and circuit complexity provably requires breakthrough circuit lower bounds! Despite this setback, one can still hope to get evidence that points at these connections, possibly by shifting some of the components of the ingredients. Namely, one may try to (a) adopt some hardness assumption, in the style of \cite{Pseopt}; (b) conclude lower bounds weaker than $\NP \not\subseteq \Ppoly$; or (c) look at non-Cook-Reckhow proof systems (such as MA proof systems or proof systems for languages beyond $\coNP$).

In this style, \citeauthor{GP18} \cite{GP18} showed that the Ideal Proof System ($\IPS$) does satisfy such a connection, to \emph{algebraic} circuit complexity. Indeed, any superpolynomial lower bound in the length of proofs in $\IPS_{\mathbb{F}}$ implies $\VP_{\mathbb{F}} \neq \VNP_{\mathbb{F}}$. \citeauthor{GP18} avoid the Pich-Santhanam barrier by means of (b) and (c) above: first, $\IPS$ is not known to be a Cook-Reckhow system, since proofs are verified by randomized machines via polynomial identity testing; second, the lower bounds are algebraic and not Boolean. Recall that while separating $\VP$ and $\VNP$ is a necessary step\footnote{Unconditionally over finite fields, and assuming the Generalized Riemann Hypothesis for infinite fields.} towards $\NP \not\subseteq \Ppoly$ \cite{Bur00}, the converse is not known.

Another interesting connection has been established in the realm of quantified Boolean formulas, where the connection can be made essentially optimal. Beyersdorff, Bonacina, Chew, and Pich \cite{BBCP20-QBF} showed that for every circuit class $\calC$, the quantified system $\calC\text{-Frege} + \forall\mathsf{red}$ is not polynomially bounded if and only if either $\PSPACE\not\subseteq \calC$ or $\calC$-Frege is not polynomially bounded. Here, $\calC\text{-Frege} + \forall\mathsf{red}$ stands for the natural quantified system obtained by extending $\calC$-Frege with a universal reduction rule, which takes care of universal quantifiers by instantiating concrete values for its variables in the hope of refuting the formula. The reason this avoids the Pich-Santhanam barrier is the disjunct in the conclusion. That is, in the context of QBF the statement of the Pich-Santhanam barrier becomes that if $\calC\text{-Frege} + \forall\mathsf{red}$ is not polynomially bounded implies $\PSPACE\not\subseteq \calC$ or $\calC$-Frege is not polynomially bounded, then it already holds that either $\NEXP\not\subseteq\Ppoly$ or $\calC$-Frege is not polynomially bounded. But this disjunction is no breakthrough, since it follows directly by a diagonalization argument anyway: if a propositional system is polynomially bounded, then $\NEXP$ is hard for $\Ppoly$ \cite{Kra04-Diag}.

\subsection*{Contributions}
We prove a new conditional connection between proof complexity and circuit complexity, giving further evidence that strong proof complexity lower bounds require circuit lower bounds. This constitutes the first example of a natural proof system that is conditionally Cook-Reckhow and whose lower bounds imply Boolean circuit lower bounds.

The system in question is (an extension of) the Implicit Extended Frege ($\iEF$) proof system of \citeauthor{Kra04} \cite{Kra04}, capable of formalizing most of contemporary complexity theory. Our result can be informally stated as follows, where 
$\iEF^{\operatorname{tt}(h)}$ stands for the proof system extending $\iEF$ by axioms $\ttable_{1/4}(h_n,2^{n/4})$ claiming there are no circuits of size $2^{n/4}$ approximating a concrete function $h$ on more than a $(1/2 + 1/{2^{n/4}})$-fraction of the inputs.\footnote{For technical reasons, we define $\iEF^{\operatorname{tt}(h)}$ using a system which is polynomially equivalent to $\iEF$ instead of $\iEF$ itself, see \Cref{def:iEF-tt}.}

\begin{theorem}
[Main theorem, informal]  \label{thm:main}
  Suppose there exists a Boolean function $h \in \NE \cap \coNE$ that is hard on average for subexponential-size circuits. If the Cook-Reckhow proof system $\iEF^{\operatorname{tt}(h)}$ is not polynomially bounded, then $\#\P \not\subseteq \FP/\poly$.
\end{theorem}

In the theorem above one could instead consider the system $\iEF^{\operatorname{tt}(h)}$ for some unconditionally hard function family $h$ that is guaranteed to exist. The only problem in this case is that we might need non-uniform advice to verify the proofs, and so the system would not be Cook-Reckhow (we refer to \citeauthor{CK07} \cite{CK07} for a systematic treatment of non-uniform proof systems).

One can interpret our theorem as improving on the connection of \citeauthor{Pseopt} \cite{Pseopt} from proof complexity to circuit complexity. Our result improves that of \citeauthor{Pseopt} by completely dropping their second assumption (the one about $\EF$ proving the existence of one-way functions under $\P \neq \NP$). The price to pay for these changes is two-fold: 
\begin{enumerate}
\item we need to replace $\EF$ by the seemingly stronger Implicit Extended Frege system ($\iEF$). Informally, $\iEF$ extends $\EF$ with an extra rule allowing us to derive a formula $\varphi$ after we have derived that a truth table of a given circuit encodes an $\EF$-proof of $\varphi$. Such a circuit is called an \emph{implicit} proof;

\item we can conclude only $\#\P\not\subseteq \FP/\poly$ from $\iEF$ lower bounds, instead of $\NP \not\subseteq \Ppoly$.
\end{enumerate}

One may also compare our result to that of \citeauthor{GP18} \cite{GP18}, who showed $\VP \neq \VNP$ (and hence hardness of computing the permanent) would follow from $\IPS$ lower bounds. Like our result, the $\IPS$ proof system is only conditionally Cook-Reckhow. Indeed, $\IPS$ is a Merlin-Arthur proof system which can be derandomized\footnote{In fact, derandomizing $\IPS$ at all by simulating it by a Cook-Reckhow system implies a non-trivial derandomization of polynomial identity testing to $\NP$ \cite{grochow2023polynomial}; this, in turn, implies some circuit lower bounds, as shown by Kabanets and Impagliazzo \cite{kabanets2004derandomizing}.} under standard assumptions, like $\E$ being hard to approximate by subexponential-size circuits. 
Our result is in some sense stronger in that the lower bounds obtained are Boolean rather than algebraic. However, we seem to be getting to lower bounds for the same problem as Grochow and Pitassi, since computing the permanent is both $\VNP$-complete and $\#\P$-complete.

We note that the requirement that $h\in\NE\cap\coNE$ is not strictly needed and, in fact, one can phrase the result in a more general style (as we do in the technical part) in which the connection holds for any extension of $\iEF$ by truth table formulas for any hard function. Observe, however, that $\iEF$ is a very strong proof system, conjectured to be strictly stronger than the standard $\EF$ and capable of formalizing most of computational complexity theory, with its bounded arithmetic counterpart being the theory $\mathsf{V^1_2}$ (or $\SOTEXP$, in the first-order setting), and so it is plausible that $\iEF$ already proves such a circuit lower bound. Indeed, the existing formalizations of complexity-theoretic statements support the assumption that $\iEF$ is able to prove efficiently practically everything we can prove in complexity theory today (or, more precisely, every $\coNP$ statement of that kind). For example, already $\EF$ can prove efficiently the PCP theorem \cite{pich2015logical}, $\AC^0$, $\AC^0[2]$ and monotone circuit lower bounds \cite{razborov1995bounded, muller2020feasibly}, or the hardness amplification producing average-case hard functions in $\E$ from worst-case hard functions in $\E$ \cite{Jer05-PhD}. Furthermore, $\iEF$ proves efficiently the correctness of Zhuk's algorithm from a CSP dichotomy \cite{Gcsp1,Gcsp2}. Hence it is plausible to imagine that if circuit lower bounds are at all provable, they may well be provable already in $\iEF$. If that turned out to be the case, then the concrete proof system in our main theorem becomes $\iEF$ itself.

\begin{corollary}(Main theorem, restated)\label{cor:main} Assume that $\iEF$ proves efficiently $\ttable_{1/4}(h_n,2^{n/4})$ for some function family $h$ and each sufficiently big $n$. Then, if $\iEF$ is not polynomially bounded, $\#\P \not\subseteq\FP/\poly$.
\end{corollary}

Let us note that one cannot make big improvements to this result without hitting the Pich-Santhanam barrier that implies $\NEXP \not\subseteq \Ppoly$ unconditionally: if we managed to prove Theorem \ref{thm:main} for a Cook-Reckhow proof system, then $\NEXP\not\subseteq\Ppoly$ would follow unconditionally. On the other hand, if our final goal is to prove $\FP\neq\#\P$, then the assumption of Theorem \ref{thm:main} is given to us for free even for some hard $h\in \E$, as otherwise, if $\E$ can be computed by subexponential-size circuits, it is not hard to show that $\P \neq \NP$ \cite{Kra04-Diag}. 

\subsection*{Consequences for self-provability of circuit upper bounds} Our result has consequences for the self-provability of circuit upper bounds. Suppose that $ \#\P\subseteq\FP/\poly$. Then, there is a sequence of polynomial-size circuits $\{C_n\}_{n\in\bbN}$ that on input a formula $\varphi$ of size $n$, outputs a satisfying assignment if one exists. This means that the propositional formula $\SAT_n(\varphi, \alpha)\rightarrow \SAT_n(\varphi,C_n(\varphi))$ claiming the correctness of $C_n$ as a SAT solver is tautological (where $\SAT_n$ is the satisfiability predicate, taking a formula $\varphi$ and an assignment $\alpha$ and evaluating the formula). But by Theorem \ref{thm:main}, $\iEF^{\operatorname{tt}}$ is now polynomially bounded, and so the proof system is able to efficiently argue for the correctness of the circuits. Namely, the mere validity of the upper bound $\#\P\subseteq\FP/\poly$ would imply the efficient propositional provability of $\SAT\in\Ppoly$.

\subsection*{Outline of the proof}
Our main result follows from a derandomization of the known fact that $\coNP \not\subseteq\MA$ implies $\#\P\not\subseteq\FP/\poly$ (see, for example, \cite[Thm. 8.22]{AB09}), together with a formalization of the underlying MA system in a suitable theory of bounded arithmetic. The implication holds, actually, for the MA system given by the sum-check protocol of Lund, Fortnow, Karloff, and Nisan \cite{LFKN92-SumCheck} in which proofs consist of circuit simulating the moves of the Prover in the protocol, so that given such a circuit, the Verifier can simulate the entire protocol on their own with the aid of randomness. If $\#\P\subseteq\FP/\poly$, then the $\#\P$-powerful Prover in the sum-check protocol can be replaced by a polynomial-size circuit and thus the system is a polynomially bounded Merlin-Arthur system. Clearly, lower bounds on the length of proofs in this system are exactly circuit lower bounds against $\#\P$.

Since $\MA$ can be derandomized under standard hardness assumptions, assuming, for example, that $\E$ is hard for subexponential-size circuits, the proof system $R$ based on the sum-check protocol above becomes a Cook-Reckhow system such that if $R$ is not polynomially bounded, then $\#\P\not\subseteq\FP/\poly$. This is almost our goal. Our task now is to replace this system by a different more standard Cook-Reckhow system $S$. This can be achieved by proving efficiently the reflection principle of the system $R$ in $S$, which essentially amounts to proving the soundness of the sum-check protocol in $S$. Here, we employ a recent work of \citeauthor{Kha23} \cite{Kha23}, in which the soundness of the sum-check protocol was formalized in $\SOTEXP$. 

In order to translate the formalization inside $\SOTEXP$ into propositional logic, we need to express the soundness of the sum-check protocol by propositional formulas. This is achieved using the machinery of approximate counting of \citeauthor{Jer07} \cite{Jer07}, which exploits Nisan-Wigderson generators based on a hard Boolean function.

\subsection*{Open problems}
Improving our result seems to require significant conceptual work. Of course, simultaneously dropping the circuit lower bound assumption as well as getting the stronger separation $\NP \not\subseteq \Ppoly$ would already imply $\NEXP \not\subseteq \Ppoly$, but one may hope to improve the existing connection by improving on one of the two fronts only. Interestingly, this seems to require progress in some of the central open questions in the theory of interactive proof systems or in hardness magnification.

\paragraph{The power of the prover.} Is it possible to strengthen the conclusion of the main theorem all the way down to $\NP\not\subseteq\Ppoly$? This would follow, for example, if we managed to design an interactive protocol for $\TAUT$ with a prover solving only $\NP$ problems and prove its correctness in $\iEF$ (unlike the current situation, where the prover is required to compute a $\#\P$-complete function). The general question of constructing a protocol for a language $L$ where the prover's power is limited to $\P^L$ is a well-known open problem in the theory of interactive proof systems (see, for example, \cite[§8.4]{AB09}).

Note, of course, that the existence of such a protocol does not suffice, since its soundness must be provable inside $\iEF$. In fact, the reason why we require $\iEF$ (or $\SOTEXP$) to carry out the formalization of the existing sum-check protocol is that one cannot feasibly talk about $\#\SAT$ directly in $\EF$ or $\SOT$  (unless $\FP = \#\P$).

\paragraph{Hardness magnification.} Is it possible to replace $\iEF$ in the main theorem by Gentzen's system $\Gent$, or even by Extended Frege? One option would be to carry out the existing formalization inside $\EF$, as mentioned above. The caveat would be, however, that we would then have to make the assumption on truth table tautologies for $\EF$. Whether $\EF$ can prove general circuit lower bounds at all seems much less believable than for $\iEF$, and so the plausibility of our hypotheses seems affected.

Instead, one may choose to keep everything in $\iEF$ and obtain the connection indirectly for $\EF$ via hardness magnification. Is there a natural class of formulas over which $\EF$ simulates $\iEF$ (and which are believably hard for $\EF$)? If so, assuming hardness of these formulas for $\EF$ would imply $\iEF$ lower bounds. By our main theorem, $\#\P \not\subseteq \Ppoly$ would follow. To the best of our knowledge, no such type of hardness magnification is known for strong proof systems.

\section{Preliminaries}
\label{sec:preliminaries}
We assume familiarity with the central concepts of computational complexity theory, propositional proof complexity and mathematical logic. Some of our work relies on formalizing standard text-book material on computational complexity in different theories of arithmetic; for the standard proofs of these results, we refer the reader to \citeauthor{AB09} \cite{AB09}. Below we review the central concepts of proof complexity and bounded theories of arithmetic and fix some notation.

\subsection{Proof complexity}
\label{subsec:proof-complexity}
Following Cook and Reckhow \cite{cookReckhow}, a \emph{propositional proof system} $S$ for the language $\TAUT$ of propositional tautologies is a polynomial-time surjective function $S : \{0,1\}^* \to \TAUT$. We shall think of $S$ as a proof checker taking as input a proof $\pi\in \{0 ,1\}^*$ and outputting $S(\pi) = \varphi$, the theorem that $\pi$ proves. Note that soundness follows from the fact that the range is exactly $\TAUT$, and implicational completeness is guaranteed by the fact that $S$ is surjective. We sometimes drop the term \emph{proof} in \emph{proof system} and use the term \emph{system} alone to refer to a function $S$ that is not guaranteed to be a Cook-Reckhow proof system (perhaps because it is unsound, or not deterministically computable).

We denote by $\size{S}{\varphi}$ the \emph{size} of the smallest $S$-proof of $\varphi$ plus the size of $\varphi$. A proof system $S$ is \emph{polynomially-bounded} if for every $\varphi \in \TAUT$, $\size{S}{\varphi} \leq |\varphi|^{O(1)}$. We say that a proof system $S$ \emph{polynomially simulates} a system $Q$, written $S \geq Q$, if for every $\varphi \in \TAUT$, $\size{S}{\varphi} \leq \size{Q}{\varphi}^{O(1)}$. Note that the notion of size and the definition of simulation do not exploit the soundness requirement of Cook-Reckhow systems. These notions are well-defined for any function whose range contains $\TAUT$. In particular, an unsound system can be polynomially bounded and simulate every other system. In some cases simulations hold only for some set $T$ of tautologies, such as the set of tautologies written as 3DNFs, and not for all formulas, and then we say that $S$ polynomially simulates $Q$ over $T$. Given a family $\{ \varphi_n\}_{n\in\mathbb{N}}$ of propositional tautologies, we write $S \vdash \varphi_n$ whenever $\size{S}{\varphi_n} \leq |\varphi_n|^{O(1)}$.

\subsubsection{Frege systems}
\label{subsec:proof-complexity-ref}
Proof complexity studies a wide variety of proof systems. The most important ones for us are \emph{Frege systems}. A Frege system is a finite set of axiom schemas and inference rules that are sound and implicationally complete for the language of propositional tautologies built from the Boolean connectives negation ($\neg$), conjunction ($\land$), and disjunction ($\lor$). A Frege proof is a sequence of formulas where each formula is obtained by either substitution of an axiom schema or by application of an inference rule on previously derived formulas. The specific choice of rules does not affect proof size up to polynomial factors, as long as there are only finitely many rules and these are sound and implicationally complete. Indeed, Frege systems polynomially simulate each other \autocite[Thm. 4.4.13]{krajicekBOOK}. Alternatively, one may choose to think of Frege systems as some variant of Natural Deduction or the Sequent Calculus for classical propositional logic.

Particularly important for us is the Extended Frege ($\EF$) system, in which proof lines can be Boolean circuits and not just formulas, which would allow in principle for more succinct proofs. We shall often consider extensions of Extended Frege by sets of additional axioms. For a set $A \subseteq \TAUT$ of tautologies recognizable in polynomial time, the system $\EF + A$ refers to Extended Frege extended with substitution instances of any formula in $A$. Note that if $A$ were to contain contingent formulas, then $\EF + A$ would not be sound; in particular, it would not be a Cook-Reckhow system, though it would be polynomially bounded.

A useful property of $\EF$ is the fact that, for every propositional system $S$, $\EF + \operatorname{Ref}_S \geq S$ \cite{KP90}. Here $\operatorname{Ref}_S$ is the sequence of tautologies encoding the \emph{reflection principle for $S$}, which states that $S$ is sound. Namely, $\operatorname{Ref}_S \coloneqq \{\operatorname{Ref}_{S, n, m}\}_{n, m \in \mathbb{N}}$ where
$\operatorname{Ref}_{S, {n, m}} \coloneqq \operatorname{Prf}_{S, n, m}(\pi, \varphi) \to \operatorname{Sat}_{n, m}(\varphi, \alpha)$,
and $\varphi$ is a formula of size $n$, $\pi$ is a purported $S$-proof of size $m$ and $\alpha$ is an assignment to the variables in $\varphi$, which are all encoded by free variables. The formula $\operatorname{Prf}_{S, n, m}$ encodes that $\pi$ is a correct $S$-proof of $\varphi$, and $\operatorname{Sat}_{n, m}(\varphi, \alpha)$ encodes the standard satisfaction relation for propositional formulas. Alternatively, one may exploit the same relation with respect to the \emph{consistency} of $S$, $\operatorname{Con}_S \coloneqq \{\operatorname{Con}_{S, m}\}_{m \in \mathbb{N}}$, where $\operatorname{Con}_{S, m} \coloneqq \neg \operatorname{Prf}_{S,1, m}(\pi, \bot)$ and $\pi$ encodes a purported proof of size $m$.

\subsubsection{Quantified propositional systems}
The focus of proof complexity is on proof systems for propositional tautologies, but it is often convenient to operate on systems capable of reasoning with \emph{quantified} Boolean formulas, where the quantification ranges over $\{0, 1\}$. We denote by $\Sigma^q_i$ (respectively, $\Pi^q_i$) the class of quantified Boolean formulas with $i$ alternations between existential and universal quantifiers, starting with an existential (respectively, universal) one. In this context, the true formulas in $\Pi_1^q$ correspond to the usual propositional tautologies.

We are particularly interested in Gentzen's Sequent Calculus for quantified propositional logic. The system extends the usual propositional Sequent Calculus by four new rules to handle quantifiers (see \cite[Def. 4.1.2]{krajicekBOOK} for a formal definition of the rules). We denote this system by $\Gent$, and by $\GentS$ its tree-like counterpart. The system $\Gent_i$, for $i\in\bbN$, corresponds to $\Gent$ where the quantified formulas appearing in the sequents can only be in the class $\Sigma_i^q \cup \Pi_i^q$. The tree-like counterpart of $\Gent_i$ is naturally denoted $\mathsf{G}^*_i$. It is useful to know that $\EF$ and $\GOS$ are polynomially equivalent with respect to $\Pi^q_1$ formulas \cite[Thm. 4.1.3]{krajicekBOOK}.

\subsubsection{Implicit proof systems}
\emph{Implicit proof systems} constitute a systematic way of obtaining, for every proof system $S$, a potentially stronger system $S'$, and were introduced by \citeauthor{Kra04} \cite{Kra04}. The essential idea is to encode a given proof in the system $S$ as a multi-output Boolean circuit taking as input a number $i$ in binary and outputting the $i$-th step of the proof. More formally, given propositional proof systems $S$ and $Q$, a proof of a tautology $\varphi$ in the \emph{implicit system} $[S, Q]$ is a pair $(\pi, C)$ consisting of a proof and a circuit, such that the truth table of $C$ encodes a valid $Q$-proof of $\varphi$ (the \emph{implicit} proof), while $\pi$ is an \emph{explicit} $S$-proof of the formula $\operatorname{Correct}_Q(\varphi, C)$, which is 
the formula stating that the truth table of $C$ is a correct $Q$-proof of $\varphi$. If $S$ and $Q$ are Cook-Reckhow proof systems, then so is $[S, Q]$.

For a system $S$, the implicit system $[S, S]$ is denoted by $\imp S$. In particular, we shall work with the \emph{Implicit Extended Frege} proof system, $\iEF \coloneqq [\EF, \EF]$. The system $\iEF$ is particularly strong, and it can in fact simulate all of $\mathsf{G}$ with respect to propositional tautologies \cite[Cor. 2.4]{Kra04}.

\subsection{Bounded arithmetic}

Our proofs extensively exploit the connections between propositional proof complexity and theories of bounded arithmetic. Below we cover the essential preliminaries needed in our formalizations, which should be accessible to any reader with basic knowledge of first-order logic. 

\subsubsection{The theories $\SOT$ and $\SOTEXP$}
Theories of bounded arithmetic capture various levels of feasible reasoning and act as a uniform counterpart of propositional systems. Intuitively, feasibility is achieved by restricting the complexity of formulas over which one can apply general reasoning schemes like induction.

The central theory for us is Buss's $\SOT$, which we think of as corresponding to polynomial-time reasoning. In this context, we work over the first-order language of bounded arithmetic, $\mathcal{L}_{\text{BA}} \coloneqq \{ 0, S, +, \cdot, <, |x|, \lfloor x/2 \rfloor, x \# y\}$, which extends the language of Peano Arithmetic by the symbols $|x|$, $\lfloor x/2 \rfloor$ and $x \# y$. The standard interpretation of $\lfloor x/2 \rfloor$ is clear. The notation $|x|$ denotes the length of the binary encoding of the number $x$, $\lceil \log (x + 1) \rceil$, while the \emph{smash symbol} $x\#y$ stands for $2^{|x|\cdot|y|}$.

The definition of \emph{bounded formulas}, is analogous to the bounded quantification one encounters in the Polynomial Hierarchy. For a quantifier $Q \in \{ \exists, \forall\}$ and a term $t$ in the language of bounded arithmetic, a formula of the form $Qx< t. \varphi(x)$ stands for either $\forall x. (x < t \to \varphi(x))$ or $\exists x. (x < t \land \varphi(x))$. These are called \emph{bounded quantifiers}. Whenever the bounded quantifier is of the form $Q < |s|$ for some term $s$, we talk about \emph{sharply bounded quantifiers}. The hierarchy of \emph{bounded formulas} consists of the classes $\Sigma_n^b$ and $\Pi_n^b$, for $n \geq 1$, which are defined by counting the alternations of bounded quantifiers ignoring the sharply bounded ones, starting with an existential (respectively, universal) one. The class $\Delta_n^b$ consists of all formulas that admit an equivalent definition in both $\Sigma_n^b$ and $\Pi_n^b$. In particular, the class $\Delta_0^b$ stands for all formulas with sharply bounded quantifiers only.

The theory $\SOT$ of \citeauthor{Buss85} \cite{Buss85} extends Robinson's arithmetic $\mathsf{Q}$ by some basic axioms for the new function symbols and the polynomial induction scheme (\textsc{PInd}) for $\Sigma_1^b$-formulas: for every $\varphi \in \Sigma_1^b$, the theory contains the axiom
\begin{equation}
    \varphi(0) \land \forall x (\varphi (\lfloor x/2 \rfloor) \to \varphi(x)) \to \forall x \varphi(x).
    \tag{\textsc{PInd}}
\end{equation}

An alternative system intended to capture polynomial-time reasoning is Cook's equational theory $\PV$ \cite{Cook75}. In the formalism of $\PV$ one has some basic function symbols and introduces new ones recursively by composition and limited recursion on notation, in the style of Cobham's functional definition of $\FP$ \cite{Cob64}. In this way, the function symbols obtained in $\PV$ are precisely those of all polynomial-time functions over the naturals. The first-order version of $\PV$ is $\PVO$ \cite{KPT91, Buss95, Cook96}. Without loss of generality, we shall work in the theory $\SOT(\PV)$, which is the theory $\SOT$ in the language of bounded arithmetic extended by all $\PV$ function symbols, meaning that we have a fresh symbol for each function in $\FP$, and induction is now available for all $\Sigma_1^b(\PV)$ formulas. We abuse notation and refer to this directly as $\SOT$.

While $\SOT$ is able to formalize a significant amount of complexity theory and some mathematics, it suffers from the drawback of being unable to even state the existence of exponentially large objects. For certain more elaborate arguments we shall work instead inside $\SOTEXP$, which patches this issue. We follow here the definition of \citeauthor{Kra04} \cite[Cor. 2.2]{Kra04}: we write $\SOTEXP \vdash \forall x \varphi(x)$ for some arithmetic formula $\varphi$ if there exists a term $t$ such that
$$\SOT \vdash \forall x \forall y (t(x) \leq |y| \to \varphi(x)).$$
The definition is somewhat indirect and may be hard to grasp at first glance. Intuitively, it allows one to derive properties about $x$ under the assumption that $y = 2^x$ exists.

The theory $\SOT$ corresponds to polynomial-time computations in the sense that the provably total relations in $\SOT$ are precisely the polynomial-time-computable ones. The same relation holds for $\SOTEXP$ and the complexity class $\EXP$.

\subsubsection{Approximate counting}
\label{subsec:apc}
Many of the formalizations carried out in bounded arithmetic require the ability to count. In some cases, small sets can be counted \emph{exactly}, but one often requires more sophisticated machinery for \emph{approximate counting}, needed to formalize many probabilistic arguments.

For $a\in \bN$, a \emph{bounded definable set} is a set of naturals $X = \{ x< a \mid \varphi(x) \} \subseteq [0, a)$, where $\varphi \in \Sigma_{\infty}^b$ is some arithmetic formula. For $X\subseteq a$ and $Y\subseteq b$, we define $X\times Y \coloneqq \{bx+y\mid x\in X, y\in Y\}\subseteq ab$ and $X\EmilU Y \coloneqq X\cup\{y+a\mid y\in Y\}\subseteq a+b$. Rational numbers are assumed to be represented by pairs of integers in the natural way. We also use the unfortunate but standard \emph{Log-notation} widespread in bounded arithmetic, by which $n\in \Log$ stands for the formula $\exists x (n = |x|)$ and $n \in \LogLog$ stands for $\exists x (n = ||x||)$.

Intuitively, from the point of view of the theory, numbers in $\Log$ are \say{small} numbers. For a circuit $C : 2^k \to 2$, where we adopt the set-theoretic custom of identifying $\{0,1\}$ with the number $2$, we can consider the bounded definable set $X_C \coloneqq \{ x<2^k \mid  C(x) = 1 \}$, and ask about the task of counting the size of $X_C$.

There exists a $\PV$-function $\operatorname{Count}(C, y) = |X_C \cap |y||$. This means that if $2^k \in \Log$, then one can do \emph{exact counting} of $|X_C|$ efficiently. We use the notation $\Pr_{x<|y|}[C(x) =1] \leq z/w$ for the $\PV$-relation $w \cdot \operatorname{Count}(C,y) \leq |y| \cdot z$.

If $2^k \not\in \Log$, exact counting becomes problematic. To avoid this, \citeauthor{Jer07} \cite{Jer05-PhD,Jer07} systematically developed the theory $\APC$ capturing probabilistic polynomial-time reasoning by means of approximate counting. The theory $\APC$ is defined as $\PVO + \dWPHP(\PV)$ where $\dWPHP(\PV)$ stands for the \emph{dual (surjective) pigeonhole principle} for all $\PV$-functions. That is, the set of all formulas
\begin{equation*}
    x>0\rightarrow\exists v<x(|y|+1).\forall u<x|y|.\ f(u)\neq v, \tag{$\dWPHP$}
\end{equation*}
where $f$ is a $\PV$-function which might involve other parameters not explicitly shown.

We write $C:X\twoheadrightarrow Y$ if $C$ is a surjective mapping from $X$ to $Y$. Let $X,Y\subseteq 2^n$ be definable sets, and $\epsilon\leq 1$. The size of $X$ is \emph{approximately less than the size of $Y$ with error $\epsilon$}, written as $X\preceq_{\epsilon} Y$, if there exists a circuit $C$, and $v\neq 0$ such that
$$C: v\times (Y\EmilU \epsilon 2^n)\twoheadrightarrow v\times X.$$
In this context, the notation $X\approx_{\epsilon}Y$ stands for $X\preceq_{\epsilon} Y$ and $Y\preceq_{\epsilon} X$. As with exact counting, the notation $\Pr_{x<y}[C(x) = 1] \circ_\epsilon z/w$ stands for $w \cdot (X_C \cap y) \circ_{\epsilon} y\cdot z$, for $\circ \in \{ \preceq, \approx\}$. Since a number $s$ is identified with the interval $[0,s)$, $X\preceq_{\epsilon} s$ means that the size of $X$ is at most $s$ with error $\epsilon$.

The definition of $X\preceq_{\epsilon} Y$ is an unbounded $\exists \Pi^b_2$ formula even if $X$ and $Y$ are defined by circuits, so it cannot be used freely in bounded induction. This problem can be solved by working in ${\sf sHARD}^{\text{A}}$, defined as the relativized theory $\SOT(\alpha)$ extended with axioms postulating that $\alpha(x)$ is a truth table of a function on $||x||$ variables hard on average for circuits of size $2^{||x||/4}$. In ${\sf sHARD}^{\text{A}}$ there is a $\PV(\alpha)$ function $\setsize$ approximating the size of any set $X\subseteq 2^n$ defined by a circuit $C$ so that $X\approx_{\epsilon} \setsize(\alpha, C,2^n,2^{\epsilon^{-1}})$ for $\epsilon^{-1}\in \Log$ (by combination of \cite[Lemma 2.14]{Jer07} and \cite[Cor. 3.6]{Jer04-Dual}).

The following key definition allows us to express that a function is indeed hard on average.

\begin{definition}[$\operatorname{Hard}^\text{A}_{\epsilon}(f)$, in $\PVO$ \cite{Jer07}] Let $f:2^k\rightarrow 2$ be a truth table of a Boolean function with $k$ inputs (with $f$ encoded as a string of $2^k$ bits, and hence with $k\in \LogLog$). We say that $f$ is \emph{average-case $\epsilon$-hard}, written as $\operatorname{Hard}^\text{A}_{\epsilon}(f)$, if for every circuit $C$ of size at most $2^{\epsilon k}$, $$|\{u<2^k\mid C(u)=f(u)\}| < (1/2+2^{-\epsilon k})2^k.$$
Note that $\operatorname{Hard}^\text{A}_{\epsilon}(f)$ is $\Pi_1^b$-definable in $\PVO$.
\end{definition}

We write $\operatorname{tt}^{\text{avg}}_\epsilon (f_k, 2^{\epsilon k}) \coloneqq || \operatorname{Hard}^\text{A}_{\epsilon}(f) ||_m$ for the propositional translation (see Section \ref{ss:proptran}) of the formula $\operatorname{Hard}^\text{A}_{\epsilon}(f)$ above, and an appropriately chosen parameter $m$ depending on $k$ and $\epsilon$.
We also consider the polynomial-time function $\tt^{\delta}(s, n, C, f)$, that checks whether $f$ is a string of length $2^n$, $C$ encodes a circuit of size at most $s$, and finally verifies whether the fraction of accepted inputs is larger than $(1/2 + 2^{-\delta n})2^n$.

The theory $\APC$ is strong enough to show that hard-on-average functions do exist.

\begin{proposition}[\citeauthor{Jer07} \cite{Jer04-Dual}] For every rational constant $\epsilon <1/3$, there exists a constant $c$ such that $\APC$ proves that for every $k\in \LogLog$ such that $k\geq c$, there exist a function $f:2^k\rightarrow 2$ that is average-case $\epsilon$-hard. \end{proposition}

The theory $\SOT$ can be relativized to $\SOT(\alpha)$. This means, in particular, that the language of $\SOT(\alpha)$, denoted also $\SOT(\alpha)$, contains symbols for all polynomial-time machines with access to the oracle $\alpha$.

\begin{definition}[${\sf sHARD}^{\text{A}}$ \cite{Jer04-Dual}] The theory ${\sf sHARD}^{\text{A}}$ is an extension of the theory $\SOT(\alpha)$ by the axioms stating
\begin{enumerate}
    \item the number $\alpha(x)$ encodes the truth table of a Boolean function in $||x||$ variables;
    \item $x\geq c\rightarrow \operatorname{Hard}^\text{A}_{1/4}(\alpha(x))$, where $c$ is the constant from the previous proposition;
    \item $||x||=||y||\rightarrow \alpha(x)=\alpha(y)$. 
\end{enumerate}
\end{definition}

The key technical tool from the framework of approximate counting is the following theorem by \citeauthor{Jer07}.

\begin{theorem}[\citeauthor{Jer07} \cite{Jer07}]\label{thm:emil-master}There is a $\PV(\alpha)$-function $\setsize$ such that ${\sf sHARD}^{\text{A}}$ proves that if $X\subseteq 2^n$ is definable by a circuit $C$, then $X\approx_{\epsilon} \setsize(\alpha, C,2^n,e)$, where $\epsilon=|e|^{-1}$.
\end{theorem} 

For a circuit $C : 2^n \to 2$, we introduce the notation
$$\Pr_{x < y}[C(x) = 1] \preceq_\epsilon^f \frac{z}{w}$$
to mean $w \cdot \setsize(f, C, 2^n, e) \leq y \cdot z$, where $\epsilon = |e|^{-1}$.

\subsubsection{Correspondences and propositional translations}\label{ss:proptran}
While our formalizations are comfortably carried out in the first-order theories presented above, we are able to transfer our results back to propositional logic thanks to the existence of \emph{propositional translations}. Following \citeauthor{krajicekBOOK} \cite{krajicekBOOK}, we say that a theory $T$ \emph{corresponds} to a propositional proof system $S$ if (i) $T$ can prove the soundness of $S$ and (ii) every universal consequence $\forall x \varphi(x)$ of $T$, where $\varphi$ is quantifier-free, admits polynomial-size proofs in $S$ when grounded into a sequence of propositional formulas. \citeauthor{Pud20} alternatively says that $S$ is the \emph{weak system} of the theory $T$ \cite{Pud20}. More formally, for such a universal formula $\varphi$, we denote by $\trans{\varphi}_n$ the propositional translation for models of size $n$. Sometimes we abuse the notation and write $\trans{\varphi}$ dropping the subscript $n$. We refer the reader to standard texts like those of \citeauthor{krajicekBOOK} \cite{krajicekBOOK} or \citeauthor{cookBOOK} \cite{cookBOOK} for formal definitions of the translation.

The key fact for us is that universal theorems of $\SOT$ admit short propositional proofs in Extended Frege. More importantly, $\SOTEXP$ corresponds to Implicit Extended Frege.

\begin{theorem}[Correspondence of $\SOTEXP$ and $\iEF$ {\cite[Thm. 2.1]{Kra04}}]
\label{thm:kra04-correspondence-iEF}
The proof system $\iEF$ corresponds to $\SOTEXP$. That is,
\begin{enumerate}[label=(\roman*)]
    \item the theory $\SOTEXP$ proves the soundness of $\iEF$;
    \item whenever a $\forall\Pi_1^b$-sentence $\forall x \varphi(x)$ is provable in $\SOTEXP$, there are polynomial-size $\iEF$-proofs of the sequence of tautologies $\{\trans{\varphi}_n\}_{n\in \bbN}$;
    \item if $\SOTEXP$ proves the soundness of some propositional system $S$, then $\iEF \geq S$.
\end{enumerate}
\end{theorem}

The translation also works for formulas beyond $\forall \Pi^{b}_1$ as long as we translate into a quantified propositional system. The definition of the translation is straightforward, and we note that $\Sigma^b_1$ consequences of $\SOT$ translated as $\Sigma_1^q$ formulas admit polynomial-size proofs in $\GOS$.

\begin{theorem}[Correspondence of $\SOT$ and $\GOS$ \cite{KP90}]
\label{thm:KP90-correspondence-S12-G1*}
Whenever a $\forall\Sigma_1^b$-sentence $\forall x \exists y \leq t. \varphi(x, y)$ is provable in $\SOT$, there are polynomial-size proofs of the sequence of $\Sigma^q_1$-formulas $\{\trans{\exists x \varphi(x, y)}_n\}_{n\in \bbN}$ in $\GOS$.
\end{theorem}

\subsection{Interactive proof systems and the sum-sheck protocol}
While our focus is on propositional proof systems in the sense of Cook and Reckhow, our work exploits relations to more lax notions of provability. Following \citeauthor{Bab85} \cite{Bab85}, an \emph{Merlin-Arthur proof system} or \emph{Merlin-Arthur protocol} for a language $L \subseteq \{ 0,1\}^*$ is a polynomial-time function $S$ together with some constant $c$ such that the two following properties are satisfied for every $x \in \{ 0,1\}^*$. Namely,
\begin{enumerate}
    \item if $x \in L$, then there exists some $\pi \in \{ 0,1\}^*$ such that $\Pr_{r \in \{ 0,1\}^{(|x| + |\pi|)^c}} [S(x, \pi, r) = 1] = 1$;
    \item if $x\not\in L$, then for every $\pi \in \{0,1\}^*$, $\Pr_{r \in \{ 0,1\}^{(|x| + |\pi|)^{c}}} [S(x, \pi, r) = 1] < 1/3$.
\end{enumerate}

The first condition formalizes \emph{completeness}, while the second corresponds to \emph{soundness}.
The complexity class $\MA$ contains all languages that admit a polynomially-bounded Merlin-Arthur protocol, meaning that there exists a constant $d$ such that the completness guarantee is strengthened to proofs $\pi \in \{0,1\}^{|x|^d}$. One should think of MA proof systems as Cook-Reckhow systems where the verifier is randomized and may thus accept some incorrect proofs with small probability.

We recall that, under the standard derandomization assumption that there exists a Boolean function family in $\E$ that is wort-case hard for subexponential-size circuits, every Merlin-Arthur system derandomizes into a Cook-Reckhow system and, in particular, $\MA = \NP$ \cite{NW94, IW97}.

Our proofs rely on a particular interactive protocol, the \emph{Sum-Check Protocol} of Lund, Fortnow, Karloff, and Nisan \cite{LFKN92-SumCheck} for the language of unsatisfiable 3CNFs. Unlike Merlin-Arthur protocols, this is an interactive protocol running for multiple rounds between a Prover and a Verifier, before the Verifier makes a decision. We now recall the details of the protocol.

\paragraph{The Sum-Check Protocol \cite{LFKN92-SumCheck}} The protocol considers a 3CNF $\varphi(x_1, \dots, x_n)$ over $m$ clauses, known to both the Verifier and the Prover.

\begin{enumerate}
    \item The Prover generates a prime number\footnote{The constant $c_p$ in the exponent comes from the formalization of the soundness of the sum-check protocol inside $\SOTEXP$ in a recent work of \citeauthor{Kha23} \cite{Kha23}; while we do not need such details in our proofs, we leave it here to be faithful to the formalization.} $p \in (2^{2n^3 + n}, 2^{{(2n^3+n)}^{c_p}}]$ together with a Pratt certificate\footnote{A \emph{Pratt certificate} is a succinct witness for primality checkable in polynomial time \cite{Pra75}. The details are not relevant for our results, but it is important that the Verifier can be convinced of $p$ being a prime.} on the primality of $p$ and sends them to the Verifier, who checks for correctness of the certificate, and aborts if incorrect.
    \item The Prover and the Verifier arithmetize $\varphi$ into a polynomial $P_\varphi(x_1, \dots, x_n)$ of degree at most $3m$ over $\mathbb{F}_p$ in the usual way: a clause like $(x \lor \neg y \lor z)$ is turned into $1 - (1-x)y(1-z)$, and one then takes the product of all such arithmetized clauses. In this way, for all $x\in \{0,1\}^n$, $\varphi(x) = 1$ if and only if $P_\varphi(x) = 1$.

    \item The Verifier sets $( a_1, \dots, a_n) \coloneqq (0, \dots, 0)$, $Q_0(a_0):=0$ and for $i \in \{ 1, \dots n\}$, the following interaction is carried out:
    \begin{enumerate}
        \item Leaving $x_i$ free, the Prover computes the coefficients of the following univariate polynomial over $\mathbb{F}_p$,
        $Q_i(x_i) \coloneqq \sum_{x_{i+1}\in\{0,1\}} \dots \sum_{x_{n}\in\{0,1\}} P_\varphi(a_1, \dots, a_{i-1}, x_i, x_{i+1}, \dots, x_n)$
        and sends the $O(m)$ coefficients of $Q_i$ to the Verifier.

        \item The Verifier checks whether $Q_i(0) + Q_i(1) = Q_{i-1}(a_{i-1})$. If the check fails, the Verifier rejects. Otherwise, it samples a random $a_{i} \in \mathbb{F}_p$ and sends it to the Prover.

        \item In the final round, instead of sending $a_n$ to the Prover, the Verifier checks whether $P_\varphi(a_1, \dots, a_n) = Q_n(a_n)$ and accepts or rejects based on this.
    \end{enumerate}
\end{enumerate}

\section{Main result}
Our proof exploits the known fact that if $\#\P \subseteq \FP/\poly$, then $\coNP \subseteq \MA$. Indeed, if $\#\P$ has small circuits one can provide polynomial-size circuits that simulate the Prover's movements in the Sum-Check protocol for $\UNSAT$, since one can consider the MA proof system in which Arthur receives from Merlin a circuit claiming to be the circuit that the Prover used to carry out their strategy, and with the aid of randomness, Arthur can execute this on his own and decide based on the outcome of this simulation.

Let us make this formal.

\begin{definition}[The $\SC$ proof system]\label{def:sc}
Let $V(p, u, \varphi, C, r)$ be the polynomial-time function carrying out the simulation of the Sum-Check protocol. Namely, $p$ is intended to be a prime in $(2^{2n^3 + n}, 2^{{(2n^3+n)}^{c_p}}]$, $u$ a Pratt certificate for $p$, $\varphi$ a 3CNF over $n$ variables, $r$ a string of random bits, and $C$ a multi-output circuit providing the Prover's responses in the interactions with the Verifier in the Sum-Check protocol.

The \emph{Sum-Check Proof System}, denoted by $\SC$, is a Merlin-Arthur proof system for proving 3DNF tautologies. An $\SC$ proof of $\varphi$ is a tuple $\langle p,u, C \rangle$ such that $p$ is indeed a prime in the interval above, correctly certified by the Pratt certificate $u$, and such that $\Pr_{r\in\bF_p^n}\big[V(p,u,\neg\varphi,C,r)=1\big]=1$.
\end{definition}

The following is just a rephrasing of the fact that $\#\P \subseteq \FP/\poly$ implies $\coNP \subseteq \MA$, in terms of the Merlin-Arthur system $\SC$.

\begin{lemma}\label{SCBounded}
    If $\#\P \subseteq \FP/\poly$, then $\SC$ is polynomially bounded over 3DNF tautologies.
\end{lemma}

\begin{proof}
    Suppose $\#\P \subseteq \FP/\poly$ and observe closely the computational tasks of the Prover in the Sum-Check protocol. On input a formula $\varphi$ over $n$ variables, the Prover sends a prime number in the range $(2^{2n^3 + n}, 2^{{(2n^3+n)}^{c_p}}]$. Note that the well-known Betrand's postulate in number theory states that for every $a > 3$, there is a prime in the interval $(a, 2a-2)$. Since $2^{{(2n^3+n)}^{c_p}} > 2 \cdot 2^{2n^3 + n} - 2$, such a prime $p$ always exists in our interval, which we fix for all our proofs of formulas over $n$ variables.
    
    We shall now argue that, on inputs of size $n$, there is a multi-output circuit $C_n$ taking as input the number of the round in the protocol and the information sent by the Verifier, and which outputs the coefficients of the polynomial $Q_i$. Note that for formulas over $m$ clauses, this is an $O(m)$-degree polynomial, and thus it suffices to evaluate it at $O(m)$ points in the field (say, the first $O(m)$ elements in $\mathbb{F}_p$) and then solve a system of linear equations to learn the coefficients. The hard task is to evaluate the polynomial $Q_i$, but this is precisely a $\#\P$ task, since it amounts to adding the outputs of the function $P_\varphi(a_1, \dots, a_{i-1}, x_i, x_{i+1}, \dots, x_n)$, which can be efficiently evaluated, for every possible $(x_{i+1}, \dots, x_n) \in \{0,1\}^{n-i}$. Since $\#\P \subseteq \FP/\poly$, there is a small circuit taking care of this task, which we use inside our circuit $C_n$. Then, the prime $p$ for inputs of size $n$, together with a suitable Pratt certificate (which is always small) and the polynomial-size circuit $C_n$ constitute a polynomial-size $\SC$ proof of the formula $\varphi$. 
\end{proof}

At this point, our goal is to extend the previous lemma from $\SC$ to a concrete and natural Cook-Reckhow system. Our goal is to do this for Implicit Extended Frege. The idea again is that $\iEF$ (or rather its first-order counterpart, $\SOTEXP$) can prove the soundness of this system and thus simulate it. We shall then derandomize the $\SC$ protocol inside $\iEF$, to argue that $\iEF$ must satisfy the same connection to lower bounds as $\SC$ does in the lemma above.

Fortunately for us, the soundness of the Sum-Check protocol was recently proven by Khaniki in the right theory of bounded arithmetic.

\begin{theorem}[Soundness of the sum-check protocol {\cite[Thm. 15.3]{Kha23-PhD}}]\label{soundness}
    There are constants $c,k\in\bN$ such that $\SOT$ proves the following sentence: for every $n,\varphi,a,p,u,C$, if it holds that (i) $\varphi$ is a 3CNF in $n$ variables where $n\geq c$, and (ii) $\varphi(a)=1$ and, (iii) $2^{2n^3+n}<p \leq 2^{{(2n^3+n)}^{c_p}}$ and,
    (iv) $n^k\in\Log\Log$,
    then $$\foo_{r\in\bF_p^n}\big[V(p,u,\varphi,C,r)=1\big]\leq \frac{n\binom{2n}{3}}{p}.$$
\end{theorem}

Based on the soundness of the interactive protocol, we can now formalize the soundness of the $\SC$ proof system from \Cref{def:sc}. The arguments that follow can be seen as a concrete application of more sophisticated techniques employed by \citeauthor{Kha23} \cite{Kha23, Kha23-PhD}, who has studied interactive protocols in the context of defining new jump operators in proof complexity. 

\begin{definition}[The $\sound_c(\SC)$ formula]\label{soundnessdef}
    We denote by $\sound_c(\SC)$ the following $\forall\Sigma^b_1$ sentence, claiming the soundness of $\SC$: for all $\varphi,a,p,u,C,f$, where $|\varphi|>c$, there is a circuit $D$ of size $\leq \lceil|f|^{1/4}\rceil$ such that if
    $$\neg\left(\foo_{r\in\bF_p^n}[V(p,u,\neg\varphi,C,r)=1]\preceq^f_\epsilon \frac{3}{8}\right)$$
    holds, then at least one of the following conditions holds:
    \begin{enumerate}[label=(\roman*)]
    \item $|f|\neq|C|^{k_a}+k'_a$ or,
        \item $\tt^{1/4}(\lceil|f|^{1/4}\rceil,||f||,D,f)=1$ or,
        \item $p\not\in (2^{2n^3+n},2^{{(2n^3+n)}^{c_p}}]$ or,
        \item $\varphi(a)=1$,
    \end{enumerate}
    where $k_a, k'_a$ are the constants from \Cref{thm:emil-master} ensuring that $\setsize$ function works properly (see the remark below), $\epsilon={1/16}$ and $n$ is the number of variables of $\varphi$. In the definition of the displayed probability, we assume that $y=p^n$, that the circuit defining the set of strings accepted by $V$ has $m$ inputs, for the smallest integer $m$ such that $2^m\ge p^n$, and that it rejects all $r\ge p^n$. 
\end{definition}

A couple of remarks are in place. First, note that even if $V$ accepts with probability 1, the probability can be approximated in Definition \ref{soundnessdef} by a significantly smaller value because of the difference between $p^n$ and $2^m$. Another relevant point is that, as a closer look at the proof of \Cref{thm:emil-master} reveals, for each $C,2^n,e$, the function $\setsize(\alpha,C,2^n,e)$ calls $\alpha$ only once. In fact, it calls $\alpha$ on an input $x$ which depends only on $|C|,n,|e|$. This is needed for the formula $\sound_c(\SC)$ to be well-defined: in Definition \ref{soundnessdef} we do not supply the Size function with an oracle generating truth tables but with a single truth table $f$ representing a single answer of the oracle.

It now suffices to verify that the encoding of the soundness of $\SC$ is indeed provable in $\SOTEXP$.

\begin{proposition}[Soundness of $\SC$ inside $\SOTEXP$]
    There is a constant $c\in\bN$ such that $\SOTEXP\vdash\sound_c(\SC)$.
\end{proposition}
\begin{proof}
    Let $c\in\bN$ be a big enough constant that can be computed from the rest of the argument and $$\sound_c(\SC)\coloneqq\forall \varphi,a,p,u,C,f\exists D\Phi(\varphi,a,p,u,C,f,D)$$ the soundness formula in \Cref{soundnessdef} above. Let $\varphi$ be a 3DNF in $n$ variables such that $|\varphi|>c$, 
 and consider $a,p,u,C,f$. Then the following cases can happen:
 \begin{enumerate}[label=(\alph*)]
     \item If $|f|\neq|C|^{k_a}+k'_a$ or $p\not\in (2^{2n^3+n},2^{{(2n^3+n)}^{c_p}}]$, then $\Phi(\varphi,a,p,u,C,f,0)$ is trivially true.
     \item If there is a circuit $D$ of size $\leq \lceil|f|^{1/4}\rceil$ such that $\tt^{1/4}(\lceil|f|^{1/4}\rceil,||f||,D,f)=1$, then the formula $\Phi(\varphi,a,p,u,C,f,D)$ is trivially true.
     \item If the previous cases do not happen and moreover $$\neg\left(\foo_{r\in\bF_p^n}[V(p,u,\neg\varphi,C,r)=1]\preceq^f_\epsilon \frac{3}{8}\right)$$ holds, then we have that $8\cdot \setsize(f, C^*,2^m,e)> 3p^n$,  where $m$ is the smallest integer such that $2^m\ge p^n$, $\epsilon\coloneqq|e|^{-1}$ and $C^*(r)\coloneqq V(p,u,\neg\varphi,C,r)$. By the assumption $\hardav(f)$ holds and by the fact that we are over $\SOT$ and we can use $f$ as a parameter in polynomial induction for $\Sigma^b_1$ formulas, we can do approximate counting using \Cref{thm:emil-master}. (Here, we use also the fact that in order to derive the conclusion of \Cref{thm:emil-master}, the axioms postulating the properties  of $\alpha(x)$ for $x$'s not queried by $\setsize$ are not needed.) Hence there is a circuit $G$ and some $v\leq \poly(m\epsilon^{-1}|C^*|)$ such that
$$G:v\times(X_{C^*}\EmilU\:\epsilon 2^m)\rra v\times\setsize(f,C^*,2^m,e).$$
As we work in $\SOTEXP$ and $G$ is surjective, we can find  a subset $A\subseteq v\times(X_{C^*}\EmilU\:\epsilon 2^m)$ such that $G$ restricted to $A$ is a one-to-one function from $A$ to $v\times\setsize(f,C^*,2^m,e)$. Now we can apply exact counting (as we have $\OEXP$) and show that
 $$\setsize(f,C^*,2^m,e)\leq |X_{C^*}|+\epsilon 2^m.$$
 By the fact that $8\cdot \setsize(f, C^*,2^m,e)> 3p^n>3\cdot 2^m/2$, we have ${2^m}/{8}<|X_{C^*}|$.
Now if $\varphi(a)=0$, by \Cref{soundness} we get
$$\foo_{r\in\bF_p^n}\big[V(p,u,\neg\varphi,\pi,r)=1\big]\leq \frac{n\binom{2n}{3}}{p}.$$
Note that $|\varphi|>c$ which implies that $n$ is big enough and as $p>2^{2n^3+n}$ we get that ${n\binom{2n}{3}}/{p}\leq {1/8}$, which implies
$$\foo_{r\in\bF_p^n}\big[V(p,u,\neg\varphi,\pi,r)=1\big]\leq {\frac{1}{8}}.$$ As $C^*$ rejects all $r\ge p^n$, this implies that $|X_{C^*}|\leq {2^m}/{8}$ which leads to a contradiction, so $\varphi(a)=1$.
 \end{enumerate}
\end{proof}

The main technical issue now is that $\sound_c(\SC)$ is a $\forall\Sigma^b_1$ sentence and thus it does not translate into a propositional formula that $\iEF$ can reason about. Instead, we shall work on a quantified propositional system, but for this to make sense we need to know the quantified propositional proof system associated with $\SOTEXP$.

We invoke the following known $\TFNP$ characterization of the $\Sigma^b_1$ consequences of $\SOTEXP$, which identifies a \say{complete} $\Sigma^b_1$ sentence $\Psi$ such that any other $\Sigma^b_1$ consequence of $\SOTEXP$ reduces to it in $\GOS$.

\begin{theorem}[\cite{Kra90, KNT11, Kra16, BS17}]\label{tfnpV}
There is a $\forall\Sigma^b_1$ sentence $\Psi\coloneqq\forall x\exists y\psi(x,y)$ (the bound on $y$ is implicit in $\psi$) such that the following statements are true:
    \begin{enumerate}[label=(\roman*)]
        \item $\SOTEXP\vdash \forall x \exists y\psi(x,y)$;
        \item for any $\forall\Sigma^b_1$ sentence $\forall x\exists y \alpha(x,y)$ such that $\SOTEXP\vdash\forall x\exists y \alpha(x,y)$, there are $\PV$ functions $f$ and $g$ such that $\SOT\vdash \forall x,y(\psi(f(x),y)\to\alpha(x,g(x,y))).$
    \end{enumerate}
\end{theorem}

In what follows, we shall work with Gentzen's system $\Gent$ extended with the propositional translation of the sentence $\Psi$ in the theorem above. We denote this system by $\GE\coloneqq\GOS+||\Psi||$ and prove the following key properties about it.

\begin{corollary}\label{wpV}
The following statements about $\GE$ hold:
\begin{enumerate}[label=(\roman*)]
    \item $\SOTEXP\vdash \sigmaref(\GE)$, i.e. the reflection principle for $\GE$ and $\Sigma^q_1$ formulas is provable in $\SOTEXP$;
    \item for every $\forall\Sigma^b_1$-sentence $\forall x\exists y\alpha(x,y)$, if $\SOTEXP\vdash\forall x\exists y\alpha(x,y)$, then there are polynomial-size $\GE$-proofs of the sequence of $\Sigma^q_1$-tautologies $\{||\exists y\alpha(y)||_n\}_{n\in\bN}$;
    \item if $\SOTEXP$ proves the soundness of a propositional proof system $S$, then $\GE\geq S$.
\end{enumerate}
\end{corollary}

\begin{proof}
    The proof of this corollary is standard and it is similar to the case of the usual correspondence for propositional proof systems and theories (see, for example, \cite{Pud20}). Here we only sketch the proof item by item.
    \begin{enumerate}[label=(\roman*)]
        \item Working in $\SOTEXP$, let $\pi$ be a $\GE$-proof of $\exists\bar{q}\varphi(\bar{p},\bar{q})$ and $a$ be an assignment for the $\bar{p}$ variables. Let $\psi'_1,...,\psi'_k$ be the substitution instances of $||\Psi||$ that are used in $\pi$. This means that there is a $\GOS$-proof $\pi'$ of the formula $\bigvee^k_{i=1}\neg\psi'_i\lor \exists\bar{q}\varphi(\bar{p},\bar{q})$. Since $\SOT$ proves the reflection principle for $\GOS$ for $\Sigma^q_1$-formulas \cite{krajicek95}, it knows that the formula is true.
        Moreover, by Theorem \ref{tfnpV} the sentence $\Psi$ is provable in $\SOTEXP$ which immediately implies that $\left(\bigvee^k_{i=1}\neg\psi'_i\right)[a/\bar{p}]$ is false and hence 
        $\exists\bar{q} \varphi(a,\bar{q})$ is true.
        \item Suppose $\SOTEXP\vdash \forall x \exists y\alpha(x,y)$ where $\forall x \exists y\alpha(x,y)$ is a $\Sigma^b_1$ sentence. Then by Theorem \ref{tfnpV}, there are $\PV$ functions $f$ and $g$ such that $\SOT\vdash \forall x,y(\psi(f(x),y)\to \alpha(x,g(x,y)))$. Then by \Cref{thm:KP90-correspondence-S12-G1*} there are polynomial-size $\GOS$-proofs of $$\{||\forall x,y(\psi(f(x),y)\to \alpha(x,g(x,y)))||_n\}_{n\in\bN}.$$ Note that $\GE$ has substitution instances of $||\Psi||$ which implies that using the rules of $\GOS$ we get polynomial-size $\GE$-proofs of the sequence $\{||\forall x\exists y\alpha(x,y)||_n\}_{n\in\bN}$.
        \item If $\SOTEXP$ proves $\refprop(S)$, then by the previous item, $\GE$ has polynomial-size proofs of the family $\{||\refprop(S)||_n\}_{n\in\bN}$ and so $\GE\geq S$ (see \Cref{subsec:proof-complexity-ref} for details on correspondences and simulations).
    \end{enumerate}
\end{proof}

Let us observe that $\GE$ is in fact equivalent to $\iEF$.

\begin{lemma}
\label{lemma:systems-equivalent}
    The proof systems $\iEF,\EF + \operatorname{Ref}_{\iEF}$ and $\GE$ are polynomially equivalent over propositional tautologies.
\end{lemma}
\begin{proof}
    By item (iii) of \Cref{wpV} and item (iii) \Cref{thm:kra04-correspondence-iEF}, $\iEF$ and $\GE$ polynomially simulate each other. As mentioned in \Cref{subsec:proof-complexity-ref}, $\EF + \operatorname{Ref}_{\iEF}\geq \iEF$. It is also easy to see that $\SOTEXP$ proves the soundness of $\EF + \operatorname{Ref}_{\iEF}$, which by item (iii) of \Cref{thm:kra04-correspondence-iEF} gives us $\iEF\geq \EF + \operatorname{Ref}_{\iEF}$.
\end{proof}

We are now ready to define the extension of $\iEF$ for which our main theorem holds. Recall that the propositional formulas $\ttt^{\text{avg}}_{1/4}(h_n,2^{n/4})$ were defined in \Cref{subsec:apc} and state the average-case hardness of a Boolean function $h_n$ represented as a truth table (hence the name tt).

\begin{definition}[The systems $\iEF^{\operatorname{tt}}$]
\label{def:iEF-tt}
Let $h = \{ h_n \}_{n \in \bbN}$ be some family of Boolean functions, and let $n_0 \in \bbN$. We denote by $\iEF^{\operatorname{tt}(h, n_0)} \coloneqq \GE + \{\ttt^{\text{avg}}_{1/4}(h_n,2^{n/4})\}_{n \geq n_0}$ the system that extends $\GE$ by the axioms claiming the hardness of $h_n$, for $n\ge n_0$.
\end{definition}

Note that $\iEF^{\operatorname{tt}(h, n_0)}$ is a family of proof systems, parameterized by a Boolean function family $h$ and some threshold parameter $n_0$. Observe that depending on the choice of $h$ and $n_0$, the  system $\iEF^{\operatorname{tt}(h, n_0)}$ may not be a Cook-Reckhow system: if $h$ is not a hard function, or $n_0$ is not large enough, we will be adding axioms which are not tautologies, and the system will be inconsistent; and even if $h$ is hard and $n_0$ is large enough, the system may require advice in order to verify the proofs. As we shall see, however, these degenerate instantiations of $\iEF^{\operatorname{tt}(h, n_0)}$ are not a problem.

What is more important, the systems $\iEF^{\operatorname{tt}(h, n_0)}$, regardless of their consistency, always simulate $\SC$.

\begin{lemma}\label{sim}
    Let $h$ be family of Boolean functions and let $n_0 \in \bbN$. The system $\iEF^{\operatorname{tt}(h, n_0)}$ polynomially simulates $\SC$ over 3DNF tautologies.
\end{lemma}
\begin{proof}
If the system $\iEF^{\operatorname{tt}(h, n_0)}$ is unsound because the added axioms are not tautologies, then the system is trivially polynomially bounded and so it simulates every other proof system.

Suppose the added axioms are indeed tautologies, meaning that the function $h$ is indeed hard on average.

      Let $\varphi_1$ be a 3DNF in $n_1$ variables and $\<p_1,u_1,C_1\>$ be a $\SC$-proof of $\varphi_1$. This means 
    $$2^{2n_1^3+n_1}<p_1\leq 2^{{(2n_1^3+n_1)}^{c_p}} \land \foo_{r\in\bF_{p_1}^{n_1}}\big[V(p_1,u_1,\neg\varphi_1,C_1,r)=1\big]=1.$$   
Note that by \Cref{soundness} and \Cref{wpV}, there are $\PV$ functions $l,g$ such that 
$$\SOT\vdash\forall \varphi,a,p,u,C,f\left(\psi(l(\varphi,a,p,u,C,f),y)\to\Phi(\varphi,a,p,u,C,f,g(\varphi,a,p,u,C,f,y))\right),$$ where $\sound_c(\SC)\coloneqq\forall \varphi,a,p,u,C,f\exists D\Phi(\varphi,a,p,u,C,f,D)$. Let $s\coloneqq|\<p,u,C\>|$. Then by \Cref{thm:KP90-correspondence-S12-G1*} there is a $s^{O(1)}$-size $\GOS$-proof of 
    $$||\forall \varphi,a,p,u,C,f\left(\psi(l(\varphi,a,p,u,C,f),y)\to\Phi(\varphi,a,p,u,C,f,g(\varphi,a,p,u,C,f,y))\right)||_{s'},$$ where $s'\coloneqq\poly(s)$. Let us rewrite the previous quantified propositional formula as $||\Psi'||\to||\Phi'||$ with the right range of parameters such that $p_1,u_1,\varphi_1,C_1$ are substituted in the formula in their corresponding places. Now we take the substitution instance $\LB_{1/4}(h_{n'},2^{n'/4})$ 
    where $|h_{n'}|\coloneqq|C_1|^{k_a}+k'_a$ and we substitute $h_{n'}$ to the variables corresponding to $f$ and therefore the disjunct which corresponds to $\tt$ disappears from $||\Phi'||$ when we 
    apply the rules of $\GOS$. Moreover, it is not hard to verify that after the substitutions  every other disjunct which corresponds to subformulas of $\sound_c(\SC)$ from  \Cref{soundnessdef} disappears except $\varphi_1$. So what we have is $\GOS$-proof of $||\Psi''||(\bar{x},\bar{y})\to\varphi_1(\bar{x})$ ($\bar{x}$ and $\bar{y}$ are disjoint variables) where $||\Psi''||$ is a substitution instance of $||\Psi'||$. 
    Since we are working in $\GE$, we have the substitution instance $\exists\bar{y}||\Psi''||(\bar{x},\bar{y})$ and therefore using the rules of $\GOS$ we get a short $\GE$-proof of $\varphi_1(\bar{x})$.
\end{proof}

Our main theorem now easily follows.

\begin{theorem}[Main theorem]
\label{thm:main-tech}
    Let $h$ be a family of Boolean functions and let $n_0 \in \bbN$. If the system $\iEF^{\operatorname{tt}(h, n_0)}$ is not polynomially bounded, then $\#\P \not\subseteq \FP/\poly$.
\end{theorem}

\begin{proof}
    By \Cref{sim} above, for every choice of $h$ and $n_0$, the system $\iEF^{\operatorname{tt}(h, n_0)}$ polynomially simulates $\SC$, so if $\iEF^{\operatorname{tt}(h, n_0)}$ is not polynomially bounded, then $\SC$ is not either. Then, by the contrapositive of \Cref{SCBounded}, $\#\P \not\subseteq \FP/\poly$.
\end{proof}

As discussed, depending on the choice of $h$ and $n_0$, the system $\iEF^{\operatorname{tt}(h, n_0)}$ may not be sound and thus possibly not a Cook-Reckhow system. However, for any fixed choice of a uniform candidate hard function, the system is concrete and exhibits the desired connection that proof complexity lower bounds for it imply strong circuit lower bounds. In particular, if there exist functions in $\NE \cap \coNE$ average-case hard for subexponential-size circuits, then we recover the version of the theorem presented in the introduction (\Cref{thm:main}).

We note that there is the possibility that $\iEF$, given its strength, already proves such strong circuit lower bounds for some Boolean function.  It is thus worth to mention the following corollary.

\begin{corollary}
    Suppose there exists a sequence of Boolean functions $\{h_n\}_{n\in\mathbb{N}}$ for which $\iEF$ has polynomial-size proofs of the formula family $\{\ttable_{1/4}(h_n, 2^{n/4})\}_{n\geq n_0}$ for some sufficiently large $n_0 \in \bbN$. If $\iEF$ is not polynomially bounded, then $\#\P \not\subseteq \FP/\poly$.
\end{corollary}

\begin{proof}
    If there is such a function $h$ and threshold $n_0$, then $\iEF^{\operatorname{tt}(h, n_0)}$ is polynomially equivalent to $\iEF$ itself, so by \Cref{thm:main-tech} the corollary follows.
\end{proof}

\subsection*{Acknowledgments}
Independently, Albert Atserias suggested to us to consider the possibility of using interactive proof systems in order to derive circuit lower bounds from proof complexity lower bounds.

We would like to thank Pavel Pudlák for useful comments and suggestions. We are also grateful to different anonymous reviewers for several comments and references.

This work was done in part while the first author was visiting the University of Oxford and the Institute of Mathematics of the Czech Academy of Sciences.

Noel Arteche was supported by the Wallenberg AI, Autonomous Systems and Software Program (WASP) funded by the Knut and Alice Wallenberg Foundation. Erfan Khaniki was supported by the Institute of Mathematics of the Czech Academy of Sciences (RVO 67985840) and GAČR grant 19-
27871X. J\'an Pich received support from the Royal Society University Research Fellowship URF$\backslash$R1$\backslash$211106 \say{Proof complexity and circuit complexity: a unified approach}.

For the purpose of Open Access, the authors have applied a CC BY public copyright license to any Author Accepted Manuscript version arising from this submission.

\printbibliography[
        heading=bibintoc
]

\end{document}